\documentclass[11pt, reqno]{amsart}
\usepackage{amsmath, amsthm, a4, latexsym, amssymb}

\def\@rmrk#1#2{\refstepcounter
    {#1}\@ifnextchar[{\@yrmrk{#1}{#2}}{\@xrmrk{#1}{#2}}}

\makeatletter\@addtoreset{equation}{section}\makeatother

 \newfont{\bfit}{cmbxti10 scaled 2000}
 \newfont{\biggi}{cmr12 scaled 2000}

 \newcommand{\eps}{\varepsilon}

 \newcommand{\R}{\mathbb{R}}
 
 \newcommand{\N}{\mathbb{N}}

 \newcommand{\prob}{\mathbb{P}}

 \newcommand{\me}{\mathbb{E}}
 
 \renewcommand{\P}{\mathbb{P}}
 
 \newcommand{\skrib}{{\mathcal B}}
 \newcommand{\skric}{{\mathcal C}}

 \newcommand{\skrie}{{\mathcal E}}
 
 \newcommand{\skrig}{{\mathcal G}}

 \newcommand{\skril}{{\mathcal L}}
 \newcommand{\skrim}{{\mathcal M}}

 \newcommand{\skrip}{{\mathcal P}}
 
 \newcommand{\skris}{{\mathcal S}}

 \newcommand{\skrix}{{\mathcal X}}

 \newcommand{\heap}[2]{\genfrac{}{}{0pt}{}{#1}{#2}}
 \newcommand{\sfrac}[2]{\mbox{$\frac{#1}{#2}$}}

\def\1{{\mathchoice {1\mskip-4mu\mathrm l}      
{1\mskip-4mu\mathrm l}
{1\mskip-4.5mu\mathrm l} {1\mskip-5mu\mathrm l}}}

\newcommand{\eq}{\begin{equation}}
\newcommand{\en}{\end{equation}}

{\nopagebreak {\hspace*{\fill}\rule{2mm}{2mm}}\\ }

{\nopagebreak {\hspace*{\fill}\rule{2mm}{2mm}}\\ }

\renewcommand{\subsection}{\secdef \subsct\sbsect}
\newcommand{\subsct}[2][default]{\refstepcounter{subsection}
\vspace{0.15cm}
{\flushleft\bf \arabic{section}.\arabic{subsection}~\bf #1  }
\nopagebreak\nopagebreak}
\newcommand{\sbsect}[1]{\vspace{0.1cm}\noindent
{\bf #1}\vspace{0.1cm}}

\newtheorem{theorem}{Theorem}[section]
\newtheorem{lemma}[theorem]{Lemma}
\newtheorem{cor}[theorem]{Corollary}

\newtheoremstyle{thm}{1.5ex}{1.5ex}{\itshape\rmfamily}{}
{\bfseries\rmfamily}{}{2ex}{}

\newtheoremstyle{rem}{1.3ex}{1.3ex}{\rmfamily}{}
{\itshape\rmfamily}{}{1.5ex}{}
\theoremstyle{rem}
\newtheorem{remark}{{\slshape\sffamily Remark}}[]

\refstepcounter{subsection}

\def\thebibliography#1{\section*{References}
  \list%
  {\arabic{enumi}.}
    {\settowidth\labelwidth{[#1]}\leftmargin\labelwidth
    \advance\leftmargin\labelsep
    \parsep0pt\itemsep0pt
    \usecounter{enumi}}
    \def\newblock{\hskip .11em plus .33em minus .07em}
    \sloppy                   
    \sfcode`\.=1000\relax}

\begin{document}
\title[LLDP  and LDP  for  Signal -to- Interference -Plus- Noise Ratio Graph  Models ]
{\Large Local  Large  Deviation   Principle,  Large  Deviation  Principle and Information Theory for the Signal -to- Interference -Plus- Noise  Ratio Graph Models}

\author[]{}

\maketitle
\thispagestyle{empty}
\vspace{-0.5cm}

\centerline{\sc By  E.  Sakyi-Yeboah$^1$,   L.  Asiedu$^1$ and  K.  Doku-Amponsah$^{1,2} $}

{$^1$Department  of  Statistics and Actuarial  Science, University of  Ghana,  BOX  LG 115, Legon,Accra}

{$^2$ Email: kdoku-amponsah@ug.edu.gh}

{$^2$ Telephone: +233205164254}

\vspace{0.5cm}

\renewcommand{\thefootnote}{}
\footnote{\textit{Acknowledgement: }  This  Research work has been  supported  by funds  from  the  Carnegie Banga-Africa Project,  University  of  Ghana}
\renewcommand{\thefootnote}{1}


\begin{quote}{\small }{\bf Abstract.}
Given devices space $D$, an  intensity measure $\lambda m\in(0,\infty)$, a transition  kernel  $Q$  from  the  space  $D$  to positive real  numbers $\R_+,$  a  path-loss function (which  depends   on  the  Euclidean distance between  the devices and   a  positive constant $\alpha$), we  define  a  Marked  Poisson Point process (MPPP).  For  a  given  MPPP  and technical constants $\tau_{\lambda},\gamma_{\lambda}:(0,\,\infty)\to (0,\infty),$  we  define a Marked Signal-to- Interference and Noise Ratio (SINR)  graph, and  associate with  it  two  empirical  measures; the  \emph{empirical  marked  measure}  and  the \emph{empirical connectivity measure}.

 For a  class of marked SINR graphs,   we  prove  a joint \emph{ large  deviation  principle}(LDP)  for  these  empirical  measures,  with speed $\lambda$  in  the  $\tau$-topology.  From  the  joint  large  deviation  principle  for  the empirical marked  measure  and  the  empirical  connectivity measure,  we  obtain  an  Asymptotic  Equipartition Property(AEP)   for  network  structured data modelled  as  a  marked  SINR  graph.  Specifically,  we  show  that  for  large dense  marked  SINR graph   one  require  approximately about  $\lambda^{2}H(Q\times Q)/\log 2$ bits to  transmit the information  contained  in  the  network  with  high  probability, where $ H(Q\times Q)$  is  a  properly  defined  entropy for  the  exponential transition  kernel  with  parameter  $c$.

 Further, we prove  a \emph {local  large  deviation principle} (LLDP)  for  the  class  of  marked SINR  graphs on  $D,$   where  $\lambda[\tau_{\lambda}(a)\gamma_{\lambda}(a)+\lambda\tau_{\lambda}(b)\gamma_{\lambda}(b)]\to \beta(a,b),$ $  a,b\in (0,\infty)$,  with  speed $\lambda$ from a \emph{ spectral potential} point.  From  the LLDP we  derive  a  conditional   LDP  for  the  marked  SINR  graphs.
 
 Note  that,  while  the joint  LDP   is  established  in  the $\tau$-topology,  the  LLDP  assume  no  topological  restriction  on  the  space  of  marked SINR  graphs. Observe also  that  all  our  rate  functions  are  expressed  in  terms  of  the \emph{relative  entropy}  or  the \emph{kullback action} or \emph{divergence function}  of  the   marked SINR  on the  devices  space  $D$
\end{quote}

\vspace{0.3cm}

\vspace{0.3cm}


\section{Introduction and Background}
Wireless ad-hoc and sensor networks have been the topic of much recent research. Now, with the introduction of 5th generation (5G) cellular systems, several techniques; including advanced multiple access technology, massive-MIMO, full-duplex, advanced modulation and coding schemes (MCSs), and simultaneous wireless information and power transfer (SWIPT) will constitute the next phase in global telecommunication standard,  see   Luo et al. ~ \cite{LSBX2019}. 5G, a  type  of  communication which  is  based on parallel processing hardware and artificial intelligence, will play a  key role in wireless networks of the next generation,  see  Bangerter et  al.~\cite{BTAS2014}.  Furthermore, the process of 5G usages  will  come along with unprecedented and exigent requirement of which connectivity is a vital cornerstone. \\

In telecommunication, wireless network comprises of a number of nodes which connect over a wireless channel.  See  Gupta and Kumar~ \cite{GK2000}. The Signal -to -Inference-Plus- Noise Ratio (SINR) determines whether a given pair of nodes can communicate with each other at a given time. Connectivity occurs in wireless network, if two nodes communicate, possibly via intermediate nodes and also, the information transport capacity of the network,  See  Ganesh and Torrisi~\cite{GT2008}. In addition, network connectivity is related to various layers, components, and metrics of wireless communication systems; however, one vital performance indicator that strongly affects other metrics as well is the signal-to-interference-plus-noise-ratio (SINR). See,  Oehmann et al.~\cite{OAVSF2015}.\\

The SINR is of key significant to the analysis and design of wireless networks. In the process of addressing the additional requirement imposed on wireless communication, in particular, a higher availability of a highly accurate modeling of the SINR is required.  Gr{\"o}nkvist and Hansson~\cite{GH2001} works on SINR model rely on the assumption that nodes are uniformly distributed in the plane. On the contrast, the complexity of  solution paves way for computational efficiency  See, example, Behzad and  Rubin~\cite{BR2003}. \\

 More so, the  SINR model can  be  made  a complex model such that each transmission is given a power and then assumes a distance-dependent path loss. A transmission is deemed to be successful if the SINR is more than some specified threshold.  See,  Amdrews \& Dinitz~\cite{AD2009}. In contrast, a lot of recent work has shown that packets are successfully received only when SINR exceeds a given threshold, and assumes that packet reception rate (PRR) is zero below this threshold. See  example,   Santi et al.~\cite{SMRDB2009}. Further study of  the  SINR  graph  model  has  shown  that  an  SINR model  of interference is a more realistic model of interference than the protocol model of interference: a receiver node receives a packet so long as the signal to interference plus noise ratio is above a certain threshold.  See,  Bakshi et  al.~\cite{BJN2017}. Furthermore, Manesh and Kaabouch~\cite{MK2017}  stated that SINR is successful if the desired receiver surpasses the threshold. This enables the transmitted signal to be decoded with satisfactory root error probability.\\

The fundamental concept of SINR model determine as transceiver design on communication system that considers interference as noise. \cite{AD2009}  examine a set of transmitter receiver pairs located in the plane with each having an associated SINR requirement; and  satisfies  as many of the requirements as possible. In all communication systems, noise generated by circuit component in the receiver is a source of signal interruption. The ratio of the signal power to noise power is termed as SINR. The SINR is a vital indicator of communication link   quality.  See Jeske and  Sampath~\cite{JS2004} . In  the  article \cite{SMRDB2009}  the wireless link scheduling problem under a graded version of the SINR interference model  is  revisited. Indeed, the article  defines  wireless link scheduling problem under the graded SINR model, where they impose an additional constraint on the minimum quality of the usable links.. \\

Li et  al.~\cite{LPNC2006} examined the statistical distribution of the SINR for the Minimum Mean Square Error (MMSE) receiver in multiple-input multiple output wireless communication.  Their study decomposed SINR model into two independent random variables; the  first  part has an exact gamma distribution and the  second  part  was shown to converge in distribution to a Normal  distribution and approximate by Generalized Gamma. Also, AIAmmouri et  al.~ \cite{AAB2017} examined the SINR   and throughput of dense cellular network with stretched exponential path loss. It  was  established  (in the article) that  the area spectral efficiency, which assumes an adaptive SINR threshold, is non-decreasing with the base station density and converges to a constant for high densities. 

 An accurate SINR estimation provides for both a more efficient system and a higher user‐perceived quality of service.\\

In this paper, we prove  the local large deviation and large deviation principles of the Signal-To-Noise and Interference Ratio graph model (SINR). In  this  sequel  we  introduce a Marked Poisson Point Process (MPPP) and  the  marked SINR  graph  model. For a  class of  the  marked SINR graph, we define the empirical marked measure and the empirical connectivity measure.  Then, we  prove  a  joint Large Deviation Principle (LDP) for the empirical marked measure and the empirical connectivity measure  of  the  marked SINR graph  model,   with speed $\lambda$ in the $\tau-$topology.  From the joint large deviation principle, we obtain an Asymptotic Equipartition Property (AEP) for network structured data modelled as an SINR graph. See, example,  Doku-Amponsah~ \cite{DA2019}  for  a generalized  version  of  the  AEP  for  wireless  sensor  networks.

 Further,  we  prove  an  LLDP  for  the  SINR  graph  and  deduce  weak  variant  of  LDP  for  the  SINR  graph models from a spectral  potential  point.   To  be specific  about  this  approach,  given  an empirical  marked  measure $\omega$,  we define  the  so-called spectral  potential $U_{R^{D}}(\omega,\,\cdot),$  for  the marked SINR  graph process,     where  $R^{D}$ is  a  properly  defined  constant  function  which  depends  on  the  device  locations and the marks.   And  we  show  that  the \emph{Kullback  action} or  the \emph{ divergence  function} $I_{\omega}(\pi),$  with  respect  to  the  empirical connectivity  measure $\pi,$  is  the  legendre  dual of  the  spectral  potential.  See,  example  Doku-Amponsah~\cite{DA2017}  for  similar  results  for the critical multitype Galton-Watson process.

\section{Statement of  Results}\label{mainresults}

\subsection{The Marked SINR Model  for  Telecommunication Networks.}

Fix  a  dimension  $d\in\N$    and  a  measureable  set  $D\subset \R^d$    with  respect  to  the  Borel-Sigma  algebra  $\skrib(\R^d).$  Denote  by  $m$  the  Lebesgues  measure on  $\R^d.$ Given  an intensity  measure,   $\lambda m:D \to [0,1]$,   a  probability  kernel  $Q$  from  $D$   to  $\R_+$,  path  loss  function   $\ell(r)=r^{-\alpha}, $    (where  $\alpha \in(0,\infty)$),  and technical constants $\tau_{\lambda}, \gamma_{\lambda}:(0\,,\,\infty)\to (0\,,\,\infty)$ we  define  the marked SINR  Graph as  follows:

\begin{itemize}
\item  We  pick  $X=(X_i)_{i\in I}$   a  Poisson  Point  Process  (PPP)  with  intensity  measure  $\lambda m:D \to [0,1]$.
\item Given   $X,$   we    assign  each  $X_i  $   a  mark  $\sigma(X_i)=\sigma_i$  independently  according  to  the  transition  kernel   $Q(\cdot \, , X_i).$
\item For   any  two  marked  points  $((X_i,\sigma_i),(X_j,\sigma_j))$  we    connect  an  edge  iff  $$ SINR(X_i,X_j, X) \ge \tau_{\lambda}(\sigma_j) \mbox{  and      $SINR(X_j,X_i, X)\ge \tau_{\lambda}(\sigma_i),$}$$  where $$SINR(X_j,X_i, X)=\frac{\sigma_i \ell(\|X_i-X_j\|)}{N_0 +\gamma_{\lambda}(\sigma_j)\sum_{i\in I\setminus\{j\} }\sigma_i\ell(\|X_i-X_j\|)}$$
\end{itemize}

We  consider  $X^{\lambda}(\mu, Q, \ell)=\Big\{[(X_i,\sigma_i), i\in I], \, E\Big\}$  under  the  joint  law  of  the  Marked  PPP  and    the  graph.  We  shall  interpret   $X^{\lambda}$    as   a  marked SINR  graph  and   $ (X_i,\sigma_i):= X_i^{\lambda}$     the  mark   of    site  $i.$  We write

 \begin{equation}
 \skris(D)=\cup_{x\subset D}\Big\{x:\,\, |x\cap A|<\infty\,\, ,\mbox{for\, any  bounded  $A\subset D$ }\Big \},
\end{equation} 
where  $|A|$  denotes  the  Cardinality  of  the  set  $A.$
 We write  $\skrix= \skris(D\times\R_+)$   and  by  $\skrim(\skrix)$  we  denote  the  space  of  positive  measures  on  the  space  $\skrix$   equipped  with  $\tau-$  topology. Henceforth,  we shall  refer  to  $\skrix$ as  locally  finite  subset  of  the  set  $D\times\R_+.$
 
For any SINR graph $X^\lambda$  we  define a probability measure, the
\emph{empirical mark measure}, ~ $L_1^{\lambda}\in\skrim(\skrix)$,~by
$$L_1^{\lambda}([x,\sigma_x]):=\frac{1}{\lambda}\sum_{i\in I}\delta_{X_i^{\lambda}}([x,\sigma_x])$$
and a symmetric finite measure, the \emph{empirical pair measure}
$L_2^{\lambda}\in\skrim(\skrix\times \skrix),$ by
$$L_2^{\lambda}([x,\sigma_x],[y,\sigma_y]):=\frac{1}{\lambda^2}\sum_{(i,j)\in E}[\delta_{(X_i^{\lambda},X_j^{\lambda})}+
\delta_{(X_j^{\lambda},X_i^{\lambda})}]([x,\sigma_x],[y,\sigma_y]).$$

 Note  that the  total mass  $\|L_1^{\lambda}\|$ of  the  empirical marked  measure  is $\1$  and  total  mass  of  
the empirical pair measure is
$2|E|/\lambda^2$.   Observe  that,  $\skrim(\skrix)\times\skrim(\skrix\times\skrix)$  is  a  closed  subset  of  $\skrim(\skrix)\times\skrim(D\times \R_+\times D\times \R_+)$  and  $$\displaystyle \P\Big\{(L_1^{\lambda},L_2^{\lambda}) \in \skrim(\skrix)\times\skrim(\skrix\times\skrix)\Big \}=1.$$
Hence, in  view  of    \cite[Lemma~4.1.5]{DZ1998}  it  is  sufficient  to  establish  Joint  LDP  for  $(L_1^{\lambda},  L_2^{\lambda})$  in  the  space  $\skrim(\skrix)\times\skrim(\skrix\times\skrix)$.   The  first  theorem  in  this section,   Theorem~\ref{main1a},  is  the  LDP  for  the  empirical  marked  measure  of  the  SINR  graph  models  in  the  space  $\skrim(\skrix).$

\begin{theorem}\label{main1a}
	Suppose   $X^{\lambda}$  is  an  SINR  graph  with  intensity  measure
	$\lambda m:D \to [0,1]$ and   a   marked  probability  kernel  $Q$  from  $D$   to  $\R_+$  and  path  loss  function   $\ell(r)=r^{-\alpha}, $  for  $\alpha>0 .$  
	Then, as $\lambda\rightarrow\infty,$
	$L_1^{\lambda}$  satisfies an  LDP   in   the  space 
	$\skrim(\skrix)$
	with good rate function
$$
\begin{aligned}
I_1(\omega)= \left\{\begin{array}{ll}H(\omega\,|m\otimes Q),  & \mbox{if  $\|\omega\|=1$ }\\
\infty & \mbox{otherwise.}		
\end{array}\right.
\end{aligned}	$$
	\end{theorem}

We  write     $\displaystyle R^{D}([x,\sigma_x],[y,\sigma_y]):=\lim_{\lambda\to\infty}\lambda R_{\lambda}^{D}([x,\sigma_x],[y,\sigma_y]),$      where  

  $$R_{\lambda}^{D} ([x,\sigma_x],[y,\sigma_y])= \int_{D} \Big[\sfrac{ \tau_{\lambda}(\sigma_x) \gamma_{\lambda}(\sigma_x)  }{\tau(\sigma_x) \gamma(\sigma_x)+(\|z\|^{\alpha}/\|x-y \|^{\alpha})} + \sfrac{ \tau_{\lambda}(\sigma_y) \gamma_{\lambda}(\sigma_y)  }{\tau_{\lambda}(\sigma_y) \gamma_{\lambda}(\sigma_y)+(\|z\|^{\alpha}/\|y-x \|^{\alpha})}\Big] dz.$$

The  next  theorem, Theorem~\ref{main1b},  is  a  conditional  LDP  for  the  empirical connectivity  measure given  the  empirical marked  measure,  and  joint  LDP   for  the  empirical  marked measure  and empirical connectivity  measure  of the  SINR  graph model.   
\begin{theorem}	\label{main1b}
		Suppose   $X^{\lambda}$  is  an  SINR  graph  with  intensity  measure
		$\lambda m:D \to [0,1]$ and   a   marked  probability  kernel  $Q$  from  $D$   to  $\R_+$  and  path  loss  function   $\ell(r)=r^{-\alpha}, $  for  $\alpha>0.$   Let  $Q$  be  the  exponential  distribution  with  parameter $c.$

		\begin{itemize}

\item[(i)]Then, as $\lambda\rightarrow\infty,$ conditional  on the  event  $L_1^{\lambda}=\omega,$  $L_2^{\lambda}$  satisfies an LDP in the  space 
$\skrim(\skrix\times\skrix)$
with  speed   $\lambda$  and  good rate function

\begin{equation}
\begin{aligned}
I_{\omega}(\pi)= \left\{\begin{array}{ll} 0,  & \mbox{if  $\pi=e^{-R^{D}}\omega\otimes\omega$  }\\
\infty & \mbox{otherwise.}		
\end{array}\right.
\end{aligned}
\end{equation} 

\item[(ii)]Then as  $\lambda\rightarrow\infty,$ the pair
$(L_1^{\lambda},\,L_2^{\lambda})$  satisfies an LDP  in  the  space
$\skrim(\skrix)\times\skrim(\skrix\times\skrix)$
with  speed   $\lambda,$  and  good rate function

\begin{equation}
\begin{aligned}
I(\omega,\,\pi)= \left\{\begin{array}{ll} H\Big(\omega\,\Big |m\otimes Q\Big),  & \mbox{if  $\pi=e^{-R^{D}}\omega\otimes\omega,$  }\\
\infty & \mbox{otherwise.}		
\end{array}\right.
\end{aligned}
\end{equation}

\end{itemize}

where $$e^{-R^{D}}\omega\otimes\omega([x,\sigma_x],[y,\sigma_y]))=e^{-R^{D}([x,\sigma_x],[y,\sigma_y])}\omega([x,\sigma_x])\omega([y,\sigma_y]).$$
\end{theorem}

 In  particular,  if we  assume   $\lambda\big[\tau_{\lambda}(\sigma_x)\gamma_{\lambda}(\sigma_x)+\tau_{\lambda}(\sigma_y)\gamma_{\lambda}(\sigma y)\big]\to \beta(\sigma_x,\sigma_y),$  for  $x\in D$ and  $\sigma_x,\sigma_y\in\R_+$  then  we  have 

$$R^{D} ([x,\sigma_x],[y,\sigma_y])= q_{\alpha} \beta(\sigma_x,\sigma_y)\|y-x \|^{\alpha} ,$$ 
where $q_{\alpha}:=\int_{D} \|z\|^{-\alpha}dz <\infty.$
Note,  $\sigma_x$  and  $\sigma_y$  are  iid  with  common  exponential distribution  $Q,$ with  parameter $c$  and  define  the  so-  called  Shannon Entropy $H$  by      
$$
H(Q\times Q)=-\int_{\skrix}\int_{\skrix}\Big[e^{-q_{\alpha} \beta(a,b) |y-x \|^{\alpha}}\log \sfrac{e^{-q_{\alpha} \beta(a,b) |y-x \|^{\alpha}}}{(1-e^{-q_{\alpha} \beta(a,b) |y-x \|^{\alpha}})}+\log(1-e^{-q_{\alpha} \beta(a,b) |y-x \|^{\alpha}})\Big] Q(da)dx\times  Q(db)dy .$$

The  next  theorem,  Theorem~\ref{main2}, is  the  Asymptotic Equipartition  Theorem  or  the  Shannon-McMillian-Breiman  Theorem  for  the  class  of  SINR  graphs

\begin{theorem}\label{main2}
	Suppose   $X^{\lambda}$  is  an  SINR  graph  with  intensity  measure  
	$\lambda m:D \to \R_+$ and   a   marked  probability  kernel  $Q$  from  $D$   to  $\R_+$  and  path  loss  function   $\ell(r)=r^{-\alpha}, $  for  $\alpha>0 .$  Assume  $\lambda\big[\tau_{\lambda}(a)\gamma_{\lambda}(a)+\tau_{\lambda}(b)\gamma_{\lambda}(b)\big]\to \beta(a,b)\in(0,\,\infty),$\, for  all  $a,b\in \R_+.$  Let  $Q$  be  the  exponential  distribution  with  parameter $c.$  Then,
	$$\lim_{\lambda \to\infty}-\frac{1}{\lambda^2}\log P(X^{\lambda})=H(Q\times Q),\qquad \mbox{with   high  probability.}$$
\end{theorem}

\begin{remark}  Theorem~\ref{main2} can  be  interpreted  as  follows:  In  order  to  code  or  transmit  the  information  contain  in  a  large  telecommunication network  modelled  as  SINR  graph  model, one  require  with  higfh  probability,  approximately  $\lambda^2 H(Q\times Q)/\log 2$  bits.
	\end{remark}

Let  $\skrig_{P}$  be  the  set  of  all  marked   SINR graphs with  intensity  measure  $\lambda m,$ where   $\lambda>0.$
  For  $\omega \in\skrim(\skrix) $ we   denote  by  $\P_{\omega}=\P\Big\{\cdot\, \Big | L_1^{\lambda}=\omega\Big\}$   and  write  $$\skrim_{\omega}=\Big\{\nu\in\skrim(\skrix \times \skrix ):\, \|\nu\|=\int_{\skrix} e^{-q_{\alpha} \beta(a,b) |y-x \|^{\alpha}}\,\omega(dx,da)\omega(dy,db)\Big \}.$$  Observe  that,  in  this  case  the  rate  function   $I_{\omega}(\pi)$  is  given  by
  
 \begin{equation}
  \begin{aligned}
   I_{\omega}(\pi)= \left\{\begin{array}{ll} 0,  & \mbox{if  $\pi=e^{-R^{D}}\omega\otimes\omega$  }\\
  \infty & \mbox{otherwise,}		
  \end{array}\right.
  \end{aligned}
  \end{equation}
  
  where  $$R^{D} ([x,\sigma_x],[y,\sigma_y])= q_{\alpha} \beta(\sigma_x,\sigma_y)\|y-x \|^{\alpha}.$$
  
Next  we  state  the  Local  large  Deviation  Principle  for  SINR graph  model  without  any  topological  restriction  on  the  space    $\skrig_P.$
  
\begin{theorem}\label{main2a}
	Suppose   $X^{\lambda}$  is  an  SINR  graph  with  intensity  measure  
	$\lambda m:D \to [0,1]$ and   a   marked  probability  kernel  $Q$  from  $D$   to  $\R_+$  and  path  loss  function   $\ell(r)=r^{-\alpha}, $  for  $\alpha>0 .$  Assume $\lambda\big[\tau_{\lambda}(a)\gamma_{\lambda}(a)+\tau_{\lambda}(b)\gamma_{\lambda}(b)\big]\to \beta(a,b)\in (0,\,\infty)$,\, for  all  $a,b\in \R_+.$  Let  $Q$  be  the  exponential  distribution  with  parameter $c.$   Then, 
	\begin{itemize}
		\item  for  any  functional  $\nu\in\skrim_{\omega}$  and    a  number $\eps>0$,  there  exists  a  weak  neighbourhood $B_{\nu}$  such  that    
		$$\P_{\omega}\Big\{X^{\lambda}\in \skrig_P\,\Big|\, L_2^{\lambda}\in B_{\nu}\Big\}\le e ^{-\lambda I_{\omega}(\pi)-\lambda\eps}.$$
		\item  for  any    $\nu\in \skrim_{\omega}$,  a  number  $\eps>o$  and  a  fine  neighbourhood $B_{\nu} $,  we  have  the  estimate: 
		$$\P_{\omega}\Big\{X^{\lambda}\in \skrig_P\,\Big|\, L_2^{\lambda}\in B_{\nu}\Big\}\ge e^{-\lambda I_{\omega}(\pi)+\lambda \eps}.$$
	\end{itemize}
	
	\end{theorem}

The  last  result, Corollay~\ref{main2b}, is  the  LDP  for  for the SINR graph  model  without  any  topological  restriction  on  the  space    $\skrig_P.$ 
\begin{cor}\label{randomg.LDM}\label{main2b}
	Suppose   $X^{\lambda}$  is  an  SINR  graph  with  intensity  measure  
	$\lambda m:D \to [0,1]$ and   a   marked  probability  kernel  $Q$  from  $D$   to  $\R_+$  and  path  loss  function   $\ell(r)=r^{-\alpha}, $  for  $\alpha>0 .$  Assume $\lambda\big[\tau_{\lambda}(a)\gamma_{\lambda}(a)+\tau_{\lambda}(b)\gamma_{\lambda}(b)\big]\to \beta(a,b)$,\, for  all  $a,b\in \R_+.$  Let  $Q$  be  the  exponential  distribution  with  parameter $c.$  
	\begin{itemize}
		\item  Let  $F$  be  closed  subset  $\skrim_{\omega}$.  Then  we  have  
		$$\limsup_{\lambda\to\infty}\frac{1}{\lambda}\log \P_{\omega}\Big\{X^{\lambda}\in \skrig_P\,\Big|\, L_2 ^{\lambda}\in F\Big\}\le -\inf_{\pi\in F}I_{\omega}(\pi).$$
		\item  Let  $O$  be  open  subset  $\skrim_{\omega}$.  Then  we  have  
		$$\liminf_{\lambda\to\infty}\frac{1}{\lambda}\log \P_{\omega}\Big\{X^{\lambda}\in \skrig_P\,\Big|\, L_2^{\lambda}\in O\Big\}\ge -\inf_{\pi\in O}I_{\omega}(\pi).$$
	\end{itemize}
	
\end{cor}

\begin{remark}We observe  from Corollary~\ref{main2b}  that  $$\lim_{\lambda\to\infty}\P_{\omega}\Big\{X^{\lambda}\in \skrig_P\,\Big|\, L_2^{\lambda}=e^{-R^{D}}\omega\otimes\omega\Big\}=1.$$
\end{remark}

\section{Proof  of  Theorem~\ref{main1a} by  Method  of  Types}

Let  $A_1,...,A_n$  be   decomposition  of  $D\times\R_{+}\subset\R^d\times \R_+.$  
We  shall  assume  henceforth that  $n<  \lambda$ and  note   by  the locally finite property of  the MPPP  that   we  have
 $$\sum_{i=1}^{n} \log\Big[\frac{e^{- \lambda m\otimes Q (A_{i})}[ \lambda m\otimes Q (A_{i})]^{\lambda \omega(A_{i})}}{[\lambda \omega(A_{i})]!}\Big]\le \log P(L^{\lambda}_{1}= \omega)\le \sum_{i=1}^{n} \log\Big[\frac{e^{- \lambda m\otimes Q (A_{i})}[ \lambda m\otimes Q (A_{i})]^{\lambda \omega(A_{i})}}{[\lambda \omega(A_{i})]!}\Big]+\eta_n ,$$
 where 
$\displaystyle\lim_{n\to\infty}\lim_{\lambda\to\infty}\sfrac{1}{\lambda}\eta_n(\lambda,A_1,...,A_n)=0$.
 The  proof  of  Lemma~\label{Types}  below  will  use  the   refined Stirling's formula
 $$(2\pi)^{\sfrac12}\lambda^{\lambda+\sfrac12}e^{-\lambda+1/(12\lambda+1)}<\lambda!<(2\pi)^{\sfrac12}\lambda^{\lambda+\sfrac12}e^{-\lambda+1/(12\lambda)}.$$

\begin{lemma}\label{Types1}
Suppose  $X^{\lambda}$  is  a  marked PPP  in a  compact set $D\times \R_+$   with  intensity  measure  $\lambda m\otimes Q$  such  that  $m $  is  absolutely  continuous   measure  on  $D.$ Then,
$$  e^{-\lambda H\big(\omega ^{(n)}\,\big|\,m^{(n)}\otimes Q^{(n)}\big)+\theta_1(\lambda)}\le \P\Big\{L_1^{\lambda}=\omega \} \le  e^{-\lambda H\big(\omega^{(n)}\,\big|\,m^{(n)}\otimes Q^{(n)}\big)+\theta_2(\lambda)}$$
$$\displaystyle\lim_{\lambda\to\infty}\theta_1(\lambda)=0,\, \lim_{\lambda\to\infty}\theta_2(\lambda)=\lim_{\lambda\to\infty}\sfrac{1}{\lambda}\eta_n(\lambda,A_1,...,A_n),$$  where  $\omega^{(n)}$  and  $m^{(n)}\otimes Q^{(n)}$  are  the  coarsening projections  of  $\omega $  and  $m\otimes Q$  on  the  decomposition  $(A_1,...,A_n).$

\end{lemma}
\begin{proof}
  For large $\lambda$,  we  have  that
$$\begin{aligned}
\log P(L_{1}^{\lambda}= \omega ) \leq \sum\{-\lambda m \otimes Q (A_{j})\}&-\log[(2\pi)^{\frac{1}{2}}(\lambda m(A_{j}))^{\lambda \omega (A_{j})+\frac{1}{2}}\exp^{-(\lambda \omega(A_{j})}]\\
&+\frac{1}{12(\lambda \omega (A_{j} )+1}+\lambda \omega (A_{j} ) \log [\lambda m \otimes Q (A_{j})]+\eta_n(\lambda,A_1,...,A_n)
\end{aligned}$$

$$\begin{aligned}
 \log P(L_{1}^{\lambda}= \omega) \le \sum\{&-\lambda m \otimes Q (A_{j})\}- \frac{1}{2} \log (2\pi)- [(\lambda \omega(A_{j}))+ \frac{1}{2}] \log [(\lambda\omega (A_{j})]\\
 &+ (\lambda \omega(A_{j}\times \Gamma_{j}))+\frac{1}{12(\lambda \omega (A_{j})+1}+\lambda \omega(A_{j}) \log\{\lambda m\otimes Q (A_{j})\}+\eta_n(\lambda,A_1,...,A_n)
\end{aligned}$$

$$ \begin{aligned}\log P(L_{1}^{\lambda}= \omega) \leq  \sum\Big \{&- \lambda [m \otimes Q(A_{j}) - \omega(A_{j}) ] - \lambda \omega(A_{j}) \log \frac{\omega(A_{j})}{m \otimes Q(A_{j})}\\
 &- \frac{1}{2} \log [\lambda m(A_{j}) ]- \frac{1}{12 [\lambda \omega(A_{j}) +1]} -\frac{1}{2} \log (2 \pi) \Big \}+\eta_n(\lambda,A_1,...,A_n)
\end{aligned} $$

$$\begin{aligned} \log P(L_{1}^{\lambda}= \omega) \leq  \sum \Big\{&- \lambda [m \otimes Q(A_{j}) - \omega(A_{j}) ] -\lambda \omega(A_{j})\log \frac{\omega(A_{j}}{m \otimes Q(A_{j})}\\
&- \lambda[\frac{\log [\lambda \omega(A_{j})] }{2 \lambda} -\frac{1}{12 \lambda^{2} \lambda \omega(A_{j}) + \lambda} + \frac{\log (2 \pi)}{2 \lambda}] \Big \}+\eta_n(\lambda,A_1,...,A_n)
\end{aligned}$$
 We choose  $\theta_{2}(\lambda)$   as
 $$ \theta_{2}(\lambda)=  \frac{\log(\lambda \omega(A_{j}))}{ 2 \lambda} -\frac{1}{12 \lambda^{2}  \omega(A_{j}) + \lambda} + \frac{\log (2 \pi)}{2 \lambda}+\eta_n(\lambda,A_1,...,A_n) $$ and  observe  that
 $$\lim_{\lambda \rightarrow \infty} \theta_{2}(\lambda)= \lim_{\lambda \rightarrow \infty}\Big [\frac{\log\lambda \omega(A_{j})}{ 2 \lambda} -\frac{1}{12 \lambda^{2}  \omega(A_{j}) + \lambda} + \frac{\log (2 \pi)}{2 \lambda}+\sfrac{1}{\lambda}\eta_n(\lambda,A_1,...,A_n)\Big]=\lim_{\lambda \rightarrow \infty}\sfrac{1}{\lambda}\eta_n(\lambda,A_1,...,A_n) $$   which   proves  the  upper bound  in the  Lemma~\ref{Types}.\\

For large  $\lambda$,  we have  the  lower bound
$$\begin{aligned}
\log P(L_{1}^{\lambda}= \omega) \geq \sum_{j=1}^{n}\{-\lambda m \otimes Q (A_{j})\}&-\log[(2\pi)^{\frac{1}{2}}(\lambda \omega(A_{j}))^{\lambda \omega (A_{j})+\frac{1}{2}}\exp^{-(\lambda \omega(A_{j})}]\\
&+\frac{1}{12(\lambda \omega (A_{j} )+1}+\lambda \omega (A_{j} ) \log [\lambda m \otimes Q (A_{j} )\}
\end{aligned}$$

$$\begin{aligned} \log P(L_{1}^{\lambda}= \omega) \geq \sum_{j=1}^{n} \{&-\lambda[m \otimes Q (A_{j})- \omega (A_{j} ] - \lambda \omega(A_{j}) \log [\lambda \omega(A_{j})]\\
&+ \lambda \omega(A_{j}) \log [\lambda m \otimes Q(A_{j})]- \frac{1}{2} \log [\lambda \omega(A_{j})]+ \frac{1}{12 [\lambda \omega (A_{j})} -\frac{1}{2} \log (2 \pi)  \}
\end{aligned} $$

$$ \begin{aligned}\log P(L_{1}^{\lambda}= \omega ) \geq  \sum_{j=1}^{n} \Big\{- \lambda [m \otimes Q(A_{j}) &- \omega(A_{j}) ] -\lambda \omega(A_{j})\log \frac{\omega(A_{j})}{m \otimes Q(A_{j})}\\
&- \lambda[\frac{\log [\lambda \omega(A_{j})] }{2 \lambda} -\frac{1}{12 \lambda^{2} \lambda \omega(A_{j})} + \frac{\log (2 \pi)}{2 \lambda}]\Big \}\end{aligned} $$

 We choose  $\theta_{1}(\lambda)$  as
 $$ \theta_{1}(\lambda)=  \frac{\log(\lambda \omega(A_{j}))}{ 2 \lambda} -\frac{1}{12 \lambda^{2}  \omega(A_{j}) } + \frac{\log (2 \pi)}{2 \lambda} ,$$   and observe that
$$\lim_{\lambda \rightarrow \infty} \theta_{1}(\lambda)= \lim_{\lambda \rightarrow \infty}\Big [\frac{\log(\lambda \omega(A_{j}))}{ 2 \lambda} -\frac{1}{12 \lambda^{2}  \omega(A_{j})} + \frac{\log (2 \pi)}{2 \lambda}\Big]=0 .$$  This   proves  the  lower  bound  of  Lemma~\ref{Types1}

\end{proof}

\begin{lemma}\label{SINR.equ2}
	Suppose   $X^{\lambda}$  is  an  SINR  graph  with  intensity  measure
	$\lambda m:D \to [0,1]$ and   a   marked  probability  kernel  $Q$  from  $D$   to  $\R_+$  and  path  loss  function   $\ell(r)=r^{-\alpha}, $  for  $\alpha>0$.   Then, for  large  $\lambda$  we  have  
	
	$$|I|\le 2 \lambda\, \mbox{ almost  surely}.$$

\end{lemma}
\begin{proof}
Note  that  $|I|$  is  expressible  as  $|I|=\sum_{k=1}^{m}I_k$,  where  $I_1, I_2, I_3,...,I_m$   are  iid poisson  random  variables  each  with  mean  $\lambda/m$  and  variance  $\lambda/m.$ Observe  that  $I_k\le a:=Vol(D),$  for  all  $k=1,2,3,...,m$   and    hence, by  applying  the  Bennett's  inequality  to  the  sequence $I_1, I_2, I_3,...,I_m$;  we  have  that
\begin{equation}\label{SINR.equ3}
\P\Big\{|I|-\me{|I|}> \lambda\Big\}\le \exp\{-\sfrac{\lambda^2}{a^2}h(a)\},
\end{equation} 

where $Vol(D)$  means  the  Volume of  the geometry  space $D$  and  $h(u)=(1+u)\log(1+u)-u.$  Now,  we use  equation~\ref{SINR.equ3}  to 
obtain  $$\P\Big\{|I|\le  \me{|I|}  +\lambda\Big\}\ge  1- \exp\{-\sfrac{\lambda^2}{a^2}h(a)\}
$$

which  gives  $$\lim_{\lambda\to\infty}\P\Big\{|I|\le  2\lambda  \Big\}\ge 1.$$

This  ends  the  proof  of  the  Lemma.

\end{proof}

Let $ \skrim_{\lambda}(\skrix) := \big\{ \omega\in \skrim(\skrix) \, : \, \lambda\omega(a) \in \N \mbox{ for all } a\in\skrix\big\}$
and  let $F$  be  a subset   of  $\skrim(\skrix).$  We write  $\beta_n:=\max(|\skrix\cap A_1|,|\skrix\cap A_2|,...,|\skrix\cap A_n|)$ and  note  that  $|\skrix\cap A_i|<\infty$, for  all $i=1,2,3,...,n,$   by  construction. We  use  Lemma~\ref{Types}  and  Lemma~\ref{SINR.equ2}   to  obtain
$$\begin{aligned}
(1+2\lambda)^{-n\beta_n}e^{-\lambda\inf_{\{\omega\in F^{o}\cap \skrim_{\lambda}(\skrix)\}} H\big(\omega^{(n)}\,\big|\,m^{(n)}\otimes Q^{(n)}\big)+\theta_1(\lambda)}&\le \sum_{\omega\in F^{o}\cap \skrim_{\lambda}(\skrix)}e^{-\lambda H\big(\omega^{(n)}\,\big|\,m^{(n)}\otimes Q^{(n)}\big)+\theta_2(\lambda)}\\
&\le \P\Big\{L_1^{\lambda}\in F\Big\}\\
 &\le \sum_{\omega\in cl(F)\cap \skrim_{\lambda}(\skrix)}e^{-\lambda H\big(\omega^{(n)}\,\big|\,m^{(n)}\otimes Q^{(n)}\big)+\theta_2(\lambda)}\\
 &\le (1+2\lambda)^{n\beta_n}e^{-\lambda\inf_{\{\omega\in cl(F)\cap \skrim_{\lambda}(\skrix)\}} H\big(\omega^{(n)}\,\big|\,m^{(n)}\otimes Q^{(n)}\big)+\theta_2(\lambda)},
\end{aligned}$$

where  $\omega^{(n)}$  and  $m^{(n)}\otimes Q^{(n)}$  are  the  coarsening projections  of  $\omega $  and  $m\otimes Q$  on  the  decomposition  $(A_1,...,A_n).$

Taking  limit  as  $\lambda\to\infty$   we  have  that  
$$\begin{aligned}\label{SINR.equ4}  
\liminf_{\lambda\to\infty}\Big\{-\inf_{\{\omega\in F^o\cap \skrim_{\lambda}(\skrix)\}} H\big(\omega^{(n)}\,\big|\,m^{(n)}\otimes Q^{(n)}\big)\Big\}&\le \lim_{\lambda\to\infty}\frac{1}{\lambda}\log\P\Big\{L_1^{\lambda}\in F\Big\}\\
&\le  \limsup_{\lambda\to\infty}\Big\{-\inf_{\{\omega\in cl(F)\cap \skrim_{\lambda}(\skrix)\}} H\big(\omega^{(n)}\,\big|\,m^{(n)}\otimes Q^{(n)}\big)\Big\}.
\end{aligned}$$

Now  we  observe  that  $cl(F)\cap \skrim_{\lambda}(\skrix)\subset cl(F) $  for  all  $\lambda\in\R_+$  and  hence  we  have  
$$\limsup_{\lambda\to\infty}\Big\{-\inf_{\{\omega\in cl(F)\cap \skrim_{\lambda}(\skrix)\}} H\big(\omega^{(n)}\,\big|\,m^{(n)}\otimes Q^{(n)}\big)\Big\}\le -\inf_{\{\omega\in cl(F)\}} H\big(\omega^{(n)}\,\big|\,m^{(n)}\otimes Q^{(n)}\big).$$ 

 Using  similar  arguments  as  \cite[Page~17]{DZ1998}    we  obtain
  
$$\liminf_{\lambda\to\infty}\Big\{-\inf_{\{\omega\in F^o\cap \skrim_{\lambda}(\skrix)\}} H\big(\omega^{(n)}\,\big|\,m^{(n)}\otimes Q^{(n)}\big)\Big\}\ge -\inf_{\{\omega\in F^o\}} H\big(\omega^{(n)}\,\big|\,m^{(n)}\otimes Q^{(n)}\big)$$

Therefore,  we   have  

 $$\begin{aligned}\label{SINR.equ5}  
-\inf_{\{\omega\in F^o\}} H\big(\omega^{(n)}\,\big|\,m^{(n)}\otimes Q^{(n)}\big) \le \lim_{\lambda\to\infty}\frac{1}{\lambda}\log\P\Big\{L_1^{\lambda}\in F\Big\}\le  -\inf_{\{\omega\in cl(F)\}} H\big(\omega^{(n)}\,\big|\,m^{(n)}\otimes Q^{(n)}\big),
 \end{aligned}$$
 where  $\omega^{(n)}$  and  $m^{(n)}\otimes Q^{(n)}$  are  the  coarsening projections  of  $\omega $  and  $m\otimes Q$  on  the  decomposition  $(A_1,...,A_n).$
Now  taking  limit  as  $n\to\infty$  we  have  

 $$\begin{aligned}\label{SINR.equ5}  
-\inf_{\{\omega\in F^o\}} H\big(\omega\,\big|\,m\otimes Q\big) \le \lim_{\lambda\to\infty}\frac{1}{\lambda}\log\P\Big\{L_1^{\lambda}\in F\Big\}\le  -\inf_{\{\omega\in cl(F)\}} H\big(\omega\,\big|\,m\otimes Q\big),
\end{aligned}$$

which  proves  the  Theorem~\ref{main1a}.

\section{ Proof  of  Theorem~\ref{main1b}  by  Gartner-Ellis  Theorem  and  the  Method of  Mixing}

Let  $A_1,...,A_n$  be  the  decomposition  of  the  space  $D\times \R_{+}.$   Note that,  for  every $(x,y)\in A_i,\, i=1,2,3,...,n,$  $\lambda L_2^{\lambda}(x,y)$  given  $\lambda L_1^{\lambda}(x)=\lambda\omega(x)$  is  binomial  with  parameters  $\lambda^2\omega(x)\omega(y)/2$  and  $p_{\lambda}(x,y).$ Let  $K$  be  the  exponential  distribution  with  parameter $c$  and  recall  that

\subsection{Proof  of   Theorem~\ref{main1b}(i)  by  Gartner-Ellis  Theorem }

 $$R_{\lambda}^{D} ([x,\sigma_x],[y,\sigma_y])= \int_{D} \Big[\sfrac{ \tau_{\lambda}(\sigma_x) \gamma_{\lambda}(\sigma_x)  }{\tau_{\lambda}(\sigma_x) \gamma_{\lambda}(\sigma_x)+(\|z\|^{\alpha}/\|x-y \|^{\alpha})} + \sfrac{ \tau_{\lambda}(\sigma_y) \gamma_{\lambda}(\sigma_y)  }{\tau_{\lambda}(\sigma_y) \gamma_{\lambda}(\sigma_y)+(\|z\|^{\alpha}/\|y-x \|^{\alpha})}\Big] dz.$$

\begin{lemma}\label{randomg.LDM1b}\label{main1c1}
	Suppose   $X^{\lambda}$  is  an  SINR  graph  with  intensity  measure
	$\lambda m:D \to [0,1]$ and   a   marked  probability  kernel  $Q$  from  $D$   to  $\R_+m$  and  path  loss  function   $\ell(r)=r^{-\alpha}, $  for  $\alpha>0$.   Then,
	
	$$ p_{\lambda}([x,\sigma_x],[y,\sigma_y])= e^{-\lambda R_{\lambda}^{D} ([x,\sigma_x],[y,\sigma_y]) }\,\qquad \mbox{and  \qquad\, $\lim_{\lambda\to\infty}\lambda R_{\lambda}^{D}([x,\sigma_x],[y,\sigma_y]))=R^{D}([x,\sigma_x],[y,\sigma_y]).$} $$.
	
\end{lemma}	
\begin{proof}
	
{\bf Calculation  of  Connectivity  Probability  by  the  Laplace Transform:}
We  note  that  the  Signal-Interference  and  Noise  Ratio is  given as  	
	$$ SINR(\tilde{X}_j,\tilde{X}_i,\tilde{ X})=\frac{\sigma_i \ell(\|X_i-X_j\|)}{N_0 +\gamma_{\lambda}(\sigma_j)\sum_{i\in I-\{j\} }\sigma_i\ell(\|X_i-X_j\|)}$$  and  the  
	 total interference  is  defined  as $$ I_{X,\sigma}(Y)= \sum_{i \epsilon I} \sigma_{i}I_i,$$   where  $I_i=\ell(\|X_i-X_j\|).$ 
	 
	 The probability that $\tilde{X}_{i}=(y,\sigma_y)$ and $\tilde{X}_{j}=(y,\sigma_y)$ are connected. $$ P(\tilde{X}_{j},\tilde{X}_{i})= P\Big[\sfrac{\sigma_i \ell(\|X_i-X_j\|)}{N_0 +\gamma_{\lambda}(\sigma_j)\sum_{i\in I-\{j\} }\ell(\|X_i-X_j\|)} \geqslant  \tau_{\lambda} (\sigma_{j})\Big] P\Big[\sfrac{\sigma_j \ell(\|X_j-X_i\|)}{N_0 +\gamma_{\lambda}(\sigma_i)\sum_{j\in I-\{i\} }\ell(\|X_j-X_i\|)} \geqslant  \tau_{\lambda} ( \sigma_i)\Big] $$ 
Now  we have  that  

	$$  P\Big[\sigma_j \ell(\|X_j-X_i\|) \geqslant\Big [(N_0 +\gamma_{\lambda}(\sigma_i)\sum_{i\in I-\{j\} }\sigma_i\ell(\|X_j-X_i\|)) \tau_{\lambda} ( \sigma_{i})\Big]$$
	$$ P(\tilde{X}_{j},\tilde{X}_{i})=  P \Big[\sigma_i \geqslant \sfrac{({N_0 +\gamma_{\lambda}(\sigma_j)\sum_{i\in I-\{j\} }\sigma_i\ell(\|X_i-X_j\|)}) \tau_{\lambda} (\sigma_{j})}{\ell(\|X_i-X_j\|)} \Big]  P \Big[\sigma_j \geqslant \sfrac{({N_0 +\gamma_{\lambda}(\sigma_i)\sum_{j\in I-\{i\} }\sigma_j\ell(\|X_j-X_i\|)}) \tau_{\lambda} (\sigma_{i})}{\ell(\|X_j-X_i\|)} \Big]$$
	Let $X_{i}= y $, $X_{j}= x $ and $I_{x,\sigma}(y)= \sum_{j\in I}\ell(\|X_j-y\|) $\\
	
	$$ 
	\begin{aligned}
	 p_{\lambda}([x,\sigma_x],[y,\sigma_y])= \Big[\int_{o}^{\infty} &P\Big( \sigma \geqslant \frac{ \tau_{\lambda}(\sigma_y) s}{\ell(\|y-x\|)}\Big ) P \Big({N_0 + \gamma_{\lambda}(\sigma_y) I_{x,\sigma}(Y) \in ds}\Big) \Big]\\
	 & \Big[\int_{o}^{\infty} P\Big( \sigma \geqslant \frac{ \tau_{\lambda}(\sigma_x) s}{\ell(\|x-y\|)}\Big ) P \Big({N_0 + \gamma_{\lambda}(\sigma_x) I_{y,\sigma}(X) \in ds}\Big) \Big] \end{aligned}$$

	Assuming that $ \sigma $ follow exponential distribution $(c)$  we  have  
	
		$$ 
	\begin{aligned}
	p_{\lambda}([x,\sigma_x],[y,\sigma_y])= \Big[\int_{o}^{\infty}  e^{- \frac{ c \tau_{\lambda}(\sigma_y) s }{\ell(\|y-x \|)} } & P \Big({N_0 + \gamma_{\lambda}(\sigma_y) I_{x,\sigma}(Y) \in ds}\Big) \Big]\\
	& \Big[\int_{o}^{\infty}  e^{- \frac{ c \tau_{\lambda}(\sigma_x) s }{\ell(\|x-y \|)} }  P \Big({N_0 + \gamma_{\lambda}(\sigma_x) I_{y,\sigma}(X) \in ds}\Big) \Big] \end{aligned}$$

	Using Laplace Transform gives \\ 
	$$ p_{\lambda}([x,\sigma_x],[y,\sigma_y])= \Big[\skril_{N_0} +\gamma_{\lambda}(\sigma_y) I_{Y,\sigma}\Big (\frac{ c \tau_{\lambda}(\sigma_y) s }{\ell(\|y-x \|)}\Big)\Big] \times \Big[\skril_{N_0} +\gamma_{\lambda}(\sigma_x) I_{X,\sigma}\Big (\frac{ c \tau_{\lambda}(\sigma_x) s }{\ell(\|x-y \|)}\Big)\Big]$$
	Since the exterior noise and interference are independent 
	$$ p_{\lambda}([x,\sigma_x],[y,\sigma_y])= \Big[\skril_{N_{0}}\Big( \frac{c \tau_{\lambda}(\sigma_y)}{\ell(\|y-x \|)}\Big) \skril_{I_{(Y,\sigma)}}\Big( \frac{c \tau_{\lambda}(\sigma_y) \gamma_{\lambda}(\sigma_y)}{\ell(\|y-x \|)}\Big)\Big]\times \Big[\skril_{N_{0}}\Big( \frac{c \tau_{\lambda}(\sigma_x)}{\ell(\|y-x \|)}\Big) \skril_{I_{(X,\sigma)}}\Big( \frac{c \tau_{\lambda}(\sigma_x) \gamma_{\lambda}(\sigma_x)}{\ell(\|y-x \|)}\Big)\Big]  $$
	Assuming there is no  external noise
	$$ p_{\lambda}([x,\sigma_x],[y,\sigma_y])= \Big[\skril_{I_{(Y,\sigma)}}\Big( \frac{c \tau_{\lambda}(\sigma_y) \gamma_{\lambda}(\sigma_y)}{\ell(\|y-x \|)}\Big)\Big]\times \Big[ \skril_{I_{(X,\sigma)}}\Big( \frac{c \tau_{\lambda}(\sigma_x) \gamma_{\lambda}(\sigma_x)}{\ell(\|y-x \|)}\Big)\Big]  $$

Hence,  by  symmetry, we  have  that	$$p_{\lambda}([x,\sigma_x],[y,\sigma_y])=p([x,\sigma_x],[y,\sigma_y])= \Big[\skril_{I_{(Y,\sigma)}}\Big( \frac{c \tau_{\lambda}(\sigma_y) \gamma_{\lambda}(\sigma_y)}{\ell(\|y-x \|)}\Big)\Big]\times \Big[ \skril_{I_{(X,\sigma)}}\Big( \frac{c \tau_{\lambda}(\sigma_x) \gamma_{\lambda}(\sigma_x)}{\ell(\|y-x \|)}\Big)\Big]$$
	Note  that    $$ \skril_{I_{(X,\sigma)}}(s)= \me(e^{-s I_{(X,\sigma)}} )  ,\mbox{for   $ s = \frac{ c \tau_{\lambda}(\sigma_x) \gamma_{\lambda}(\sigma_x)}{\ell(\|y-x \|)}$}.$$
	$$ \skril_{I_{(X,\sigma)}} (s)= \exp \Big\{ \int_{D} \int_{0}^{\infty} \big [e^{-s \sigma { \ell (\|z \|)} } -1\big] Q(d\sigma,x) \mu(dz) \Big\} $$
	Let $ \mu(dz) = \lambda dz $ and  recall that  the  battery  is  assumed  to  be    $ Q(d\sigma,x) = c e^{-c \sigma}$
	$$ \skril_{I_{(X,\sigma)}}(s)= \exp \Big\{ \int_{D} \int_{0}^{\infty}  [e^{-s \sigma {\ell (\|z \|)} } -1] c e^{-c \sigma} d\sigma \lambda dz \Big \} $$
	$$ \skril_{I_{(X,\sigma)}}(s)= \exp \Big\{ \lambda \int_{D} \int_{0}^{\infty}  [c e^{-s \sigma {\ell(\|z \|)}-c \sigma } - c e^{-c \sigma} d\sigma]  dz\Big \} $$

	$$ \skril_{I_{(X,\sigma)}}= \exp \Big\{ \lambda \int_{D} [c \int_{0}^{\infty}   e^{- \sigma[s {\ell (\|z \|)}+c]} - \int_{0}^{\infty}c e^{-c \sigma} d\sigma]  dz\Big \} $$

	$$ \skril_{I_{(X,\sigma)}}(s)= \exp\Big \{ \lambda \int_{D} [c \frac{1}{s {\ell (\|z \|)}+c]} - 1] dz\Big \} $$
	$$ \skril_{I_{(X,\sigma)}}(s)= \exp\Big \{ \lambda \int_{D} \frac{-s {\ell(\|z \|) } }{s {\ell (\|z \|)+c}} dz\} $$

	$$ p_{\lambda}([x,\sigma_x],[y,\sigma_y])= \exp\Big \{-\lambda \int_{0}^{\infty} \frac{s {\ell(\|z \|) } }{s {\ell (\|z \|)+c}} dz   -\lambda \int_{0}^{\infty} \frac{t {\ell(\|z \|) } }{t {\ell (\|z \|)+c}} dz  \Big \}$$
	By substitution, $s = \frac{ c \tau_{\lambda}(\sigma_x) \gamma_{\lambda}(\sigma_x)}{\ell(\|x-y \|)} $  and  $t = \frac{ c \tau_{\lambda}(\sigma_y) \gamma_{\lambda}(\sigma_y)}{\ell(\|y-x \|)} $
	$$ p_{\lambda}([x,\sigma_x],[y,\sigma_y])= \exp\Big \{-\lambda \int_{D} \sfrac{\frac{c \tau_{\lambda}(\sigma_x) \gamma_{\lambda}(\sigma_x) {\ell (\|z \|)} }{\ell(\|x-y \|)} }{\frac{c\tau_{\lambda}(\sigma_x) \gamma_{\lambda}(\sigma_x)}{\ell(\|x-y \|)}\ell (\|z \|)+c} dz-\lambda \int_{D} \sfrac{\frac{c \tau_{\lambda} \gamma_{\lambda} {\ell (\|z \|)} }{\ell(\|y-x \|)} }{\frac{c\tau_{\lambda} \gamma_{\lambda}}{\ell(\|y-x \|)}\ell (\|z \|)+c} dz\Big \} $$
	Using $\ell (r) = r^{-\alpha}$  we  obtain the  expression
		$$ p_{\lambda}([x,\sigma_x],[y,\sigma_y])= \exp\Big \{-\lambda \int_{D} \sfrac{ \tau_{\lambda}(\sigma_x) \gamma_{\lambda}(\sigma_x)  }{\tau_{\lambda}(\sigma_x) \gamma_{\lambda}(\sigma_x)+(\|z\|^{\alpha}/\|x-y \|^{\alpha})} dz-\lambda\int_{D} \sfrac{ \tau_{\lambda}(\sigma_y) \gamma_{\lambda}(\sigma_y)  }{\tau_{\lambda}(\sigma_y) \gamma_{\lambda}(\sigma_y)+(\|z\|^{\alpha}/\|y-x \|^{\alpha})} dz\Big \} $$

We  write  $$R_{\lambda}^{D} ([x,\sigma_x],[y,\sigma_y])= \int_{D} \Big[\sfrac{ \tau_{\lambda}(\sigma_x) \gamma_{\lambda}(\sigma_x)  }{\tau_{\lambda}(\sigma_x) \gamma_{\lambda}(\sigma_x)+(\|z\|^{\alpha}/\|x-y \|^{\alpha})} + \sfrac{ \tau_{\lambda}(\sigma_y) \gamma_{\lambda}(\sigma_y)  }{\tau_{\lambda}(\sigma_y) \gamma_{\lambda}(\sigma_y)+(\|z\|^{\alpha}/\|y-x \|^{\alpha})}\Big] dz.$$
 and  observe   that  we  have  
	$$ p_{\lambda}([x,\sigma_x],[y,\sigma_y])= e^{-\lambda R_{\lambda}^{D} ([x,\sigma_x],[y,\sigma_y]) }.$$   and   Therefore,   we  have  $$\lim_{\lambda\to\infty}\lambda R_{\lambda}^{D}([x,\sigma_x],[y,\sigma_y])=R^{D}([x,\sigma_x],[y,\sigma_y]). $$	  which  completes  the  proof  of   Lemma~\ref{main1c1}.
		\end{proof}

{\bf Computation of  the  log moment generation function }\\

\begin{lemma}\label{randomg.LDM1b}\label{main1c}
	Suppose   $X^{\lambda}$  is  an  SINR  graph  with  intensity  measure
	$\lambda Leb(x):D \to [0,1]$ and   a   marked  probability  kernel  $Q$  from  $D$   to  $\R_+m$  and  path  loss  function   $\ell(r)=r^{-\alpha}, $  for  $\alpha>0 ,$  conditional  on the  event  $L_1^{\lambda}=\omega.$  Let $ g:\skrix\times \skrix\to \R$ be  bounded  function.  Then,
	
	$$\begin{aligned}\lim_{\lambda\to\infty}\frac{1}{\lambda}\log\me \Big\{e^{\lambda\langle g, \, L_2^{\lambda}\rangle/2 }\Big | L_1^{\lambda}=\omega\Big\}&=\frac{1}{2}\lim_{n\to\infty}\Big[\sum_{j=1}^{n}\sum_{i=1}^{n}\int_{y\in A_j}\int_{x\in A_i}g(x,y)e^{-R^{D}(x,y)}\omega(dx)\omega(dy)\Big]\\
	&=\frac{1}{2}\int_{\skrix}\int_{\skrix} g(x,y)e^{-R^{D}(x,y)}\omega(dx)\omega(dy).
	\end{aligned}$$
\end{lemma}

\begin{proof}	
Now  we  observe  that 
	$$ \me\Big \{ e^{\int \int \lambda g(x,y) L_{2}^{\lambda}(dx,dy)/2} \mid L_{1}^{\lambda}=\omega\Big \}= \me\Big\{\prod_{x \in D} \prod_{y \in D} e^{g(x,y)\lambda L_{2}^{\lambda}(dx,dy)/2} \Big\}$$
	
	$$ \me\Big \{\prod_{x \epsilon D} \prod_{y \epsilon D} e^{g(x,y)\lambda L_{2}^{\lambda}(dx,dy/2)} = \prod_{i=1} \prod_{j=1} \prod_{x \epsilon A_{i}} \prod_{y \epsilon A_{j}} \me\Big\{e^{ \frac{g(x,y)}{\lambda} \lambda^{2}L_{2}^{\lambda}(dx,dy)/2  }\Big \} $$
Hence  by   from  Lemma~\ref{main1c1} we  have	
	$$ \log \Big \{e^{\lambda \langle g,L_{2}^{\lambda}\rangle/2}\Big|L_1^{\lambda}=\omega \Big\} = \sum_{j=1}^{n} \sum_{i=1}^{n}\int_{A_j}\int_{A_i}\log \Big[1-p_{\lambda}(x,y)+p_{\lambda}(x,y)e^{\frac{g(x,y)}{\lambda}}\Big]^{\lambda^{2} \omega \otimes \omega (dx,dy)/2} $$

	By Euler's Formula, see  example~\cite[pp.~1998]{DM2010},  we  have
	$$   \frac{1}{\lambda} \log E \{e^{\lambda \langle g,L_{2}^{\lambda}\rangle/2 } \mid L_{1}^{\lambda}=\omega \} = \frac{1}{\lambda}\sum_{j=1} \sum_{i=1} \int_{A_{i}} \int_{A_{j}} \log\Big [1+ \frac{g(x,y)}{\lambda} p_{\lambda}(x,y) + o (\lambda^{2} )\Big ]^{ \lambda^2 \omega \otimes \omega (dx,dy)/2} $$
	$$  \frac{1}{\lambda} \log \me \{e^{\lambda \langle g,L_{2}^{\lambda}\rangle /2} \mid L_{1}^{\lambda}=\omega \} = \lim_{\lambda \rightarrow \infty } \sum_{j=1} \sum_{i=1}\int_{A_{i}} \int_{A_{j}} \log \Big [1+ \frac{g(x,y)
	}{\lambda} p_{\lambda}(x,y) + o (\lambda^{2} ) \Big]^{ \lambda \omega \otimes \omega (dx,dy)/2} $$
	
	$$\lim_{\lambda \rightarrow \infty} \frac{1}{\lambda} \log \me \Big\{e^{\lambda \langle g,L_{2}^{\lambda}\rangle /2} \mid L_{1}^{\lambda}=\omega\Big \} = \frac{1}{2} \sum_{j=1} \sum_{i=1} \int_{A_{i}} \int_{A_{j}}\log\Big [e^{g(x,y) e^{-R^{D}(x,y)}}\Big]\omega \otimes \omega (dx,dy) $$
	
	$$\lim_{\lambda \rightarrow \infty}\frac{1}{\lambda} \log \me\{e^{\lambda \langle g,L_{2}^{\lambda}\rangle/2} \mid =\omega \} 
	= \frac{1}{2} \sum_{j=1} \sum_{i=1} \int_{A_{i}} \int_{A_{j}} g(x,y) e^{-R^{D}(x,y)}\omega \otimes \omega (dx,dy) $$
	
	$$ \begin{aligned}
	\lim_{n \rightarrow \infty} \lim_{\lambda \rightarrow \infty} \frac{1}{\lambda} \log \me \{e^{\lambda \langle g,L_{2}^{\lambda}\rangle/2 } \mid L_{1}^{\lambda}=\omega \}& =\frac{1}{2} \lim_{n \rightarrow \infty } \sum_{j=1} \sum_{i=1} \int_{A_{i}} \int_{A_{j}}  [g(x,y) e^{-R^{D}(x,y)}\omega \otimes \omega (dx,dy)]\\
	&	=\frac{1}{2}\int_{\skrix} \int_{\skrix} g(x,y) e^{-R^{D}(x,y)}\omega \otimes \omega (dx,dy)  
\end{aligned}$$	
\end{proof}

Hence,	by  Lemma~\ref{main1c}  and the  Gartner-Ellis theorem, $L_{2}^{\lambda}$   conditional  on  $L_{1}^{\lambda}= \omega$  obey a  large  deviation  principle with speed and  rate function
		$$ I_{\omega}(\pi) = \frac{1}{2}\sup_{g} \Big\{ \int_{\skrix} \int_{\skrix} g(x,y) \pi(dx,dy)- \int_{\skrix} \int_{\skrix} g(x,y) e^{R^{D}(x,y)}\omega \otimes \omega(dx,dy) \Big\}$$ 
	
	 which 	clearly   reduces  to  the rate function  given by 
	
	\begin{equation}
	\begin{aligned}
	I_{\omega}(\pi)= \left\{\begin{array}{ll}  0 & \mbox{if  $ \pi= e^{-R^{D}}\omega \otimes \omega $  }\\
	\infty & \mbox{otherwise.}		
	\end{array}\right.
	\end{aligned}
	\end{equation}

\subsection{ Proof of  Theorem\ref{main1a}(ii)   by  Method  of  Mixtures.}

For any $\lambda\in\R_+$ we define
$$\begin{aligned}
\skrim_{\lambda}(\skrix) & := \big\{ \omega\in \skrim(\skrix) \, : \, \lambda\omega(a) \in \N \mbox{ for all } a\in\skrix\big\},\\
 \skrim_{\lambda}(\skrix\times\skrix) & := \big\{ \pi\in
\tilde\skrim(\skrix\times\skrix) \, : \, 
\lambda\,\pi(a,b) \in \N,\,  \mbox{ for all } \, a,b\in\skrix\times\skrix
\big\}\, .
\end{aligned}$$

We denote by
$\Theta_{\lambda}:=\skrim_{\lambda}(\skrix)$
and
$\Theta:=\skrim(\skrix)$.
With
$$\begin{aligned}
P_{ \omega_{\lambda}}^{(\lambda)}(\eta_\lambda) & := \prob\big\{L_2^{\lambda}=\eta_{\lambda} \, \big| \, L_1^{\lambda}=\omega_{\lambda}\big\}\, ,\\
P^{(\lambda)}(\omega_\lambda) & :=
\prob\big\{L_1^{\lambda}=\omega_\lambda\big\}
\end{aligned}$$

the joint distribution of $L_1^{\lambda}$ and $L_2^{\lambda}$ is
the mixture of $P_{ \omega_{\lambda}}^{(\lambda)}$ with
$P^{(\lambda)}(\omega_{\lambda})$ defined as
\begin{equation}\label{randomg.mixture}
d\tilde{P}^{\lambda}( \omega_{\lambda}, \eta_\lambda):= dP_{	\omega_{\lambda}}^{(\lambda)}(\eta_{\lambda})\, dP^{(\lambda)}( \omega_{\lambda}).\,
\end{equation}


The following lemmas ensure the validity of
large deviation principles for the mixtures and for the goodness of
the rate function if individual large deviation principles are
known. See for  example, \cite[Page~30]{DM2010}  and  the  references therein.  We  observe  that the  family of
measures $({P}^{\lambda} \colon \lambda\in(0,\infty))$  is  exponentially tight on
$\Theta.$

\begin{lemma}[] \label{Tigtness} The  family of
	measures $(\tilde{P}^{\lambda} \colon \lambda\in\R_+)$  is  exponentially tight on
	$\Theta\times\skrim(\skrix\times\skrix).$
\end{lemma}

\begin{proof} 
	
	Let  $\displaystyle\eta>\min_{a,b}R^{D}(a,b)>0$  and  $t=1-(1-e^{-1})e^{-\eta}.$ Then,  we  use  Chebysheff's  inequality  and  Lemma~\ref{main1c}, to obtain  (for  sufficiently  large  $\lambda$),
	
	$$\P\Big\{  |E|\ge \lambda^2 l\Big\}\le e^{-\lambda^2 l}\me\{e^{|E|}\}\le e^{-\lambda^2 l }\sum_{i=0}^{\infty}\sum_{k=0}^{i} e^{k}\Big(\heap{i}{k}\Big)\Big(e^{-\eta}\Big)^{k}\Big(1-e^{-\eta}\Big)^{i-k}\frac{e^{-\lambda}\lambda^{i}}{i!}\le e^{-\lambda^2 l}e^{-\lambda}e^{t\lambda}.$$
	Given  $N\in\N$  we  choose $N>q$  and  observe  that  for  sufficiently  large  $\lambda$   we  have
	
	$$\P\Big\{ |E|\ge \lambda^2 N\Big\} \le  
	e^{-\lambda^2 q} .$$
	Therefore,  we  have  
	$$\P\Big\{ \|L_2^{\lambda}\|\ge \lambda^2 N/2\Big\} \le  
	e^{-\lambda^2 q/2} ,$$
	
	which  establishes   Lemma~\ref{Tigtness}.

\end{proof}

Define the function
$I\colon{\Theta}\times\skrim(\skrix\times\skrix)\rightarrow[0,\infty],$  by

\begin{equation}\label{main1b.equ}
\begin{aligned}
I(\omega,\pi)= \left\{\begin{array}{ll} H\Big(\omega\,\big|m\otimes Q\Big),  & \mbox{if  $\pi=e^{-R^{D}}\omega\otimes\omega$  }\\
\infty & \mbox{otherwise.}		
\end{array}\right.
\end{aligned}
\end{equation}

\begin{lemma}[]\label{SINR.convexgoodrate}
	$I$ is lower semi-continuous.
\end{lemma}

\begin{proof}
Let $(\omega,\pi)\in{\Theta}\times\skrim(\skrix\times\skrix)$   and  observe  that   $\pi=e^{-R^{D}}\omega\otimes\omega$ is  closed  condition.  Further, we  note  that  the relative  entropy,   $ H\Big(\omega\,\big|m\otimes Q\Big),$ is  a  lower  semi-continuous  function  on  the  space $ {\Theta}\times\skrim(\skrix\times\skrix) $.   As  $I$  is  a  function  of  a  relative entropy, we   conclude  that   	$I$ is lower semi-continuous.

\end{proof}

Using \cite[Theorem~5(b)]{Bi2004}   together with the two previous lemmas and the
large deviation principles we have established
Theorem~\ref{main1a}  and  Theorem~\ref{main1b}(i) ensure
that under $(\tilde{P}^\lambda)$ the random variables $(L_1^{\lambda}, L_2^{\lambda})$ satisfy a large deviation principle on
$\skrim(\skrix) \times \skrim(\skrix\times\skrix)$ with good rate function  $I$  which  ends  the  proof of  Theorem~\ref{main1b}(ii).

\section{Proof of   Theorem~\ref{main2} by Large  deviation  Technique }

\subsection{Proof  of  Theorem~\ref{main2}}
We  begin the  proof  of  the  asymptotic equipartition property,  by  first  establishing  a  weak  law  of large  numbers  for   the empirical  mark  measure  and  the  empirical  pair  measure  og  the  SINR  graph.
\begin{lemma}\label{WLLN}
	Suppose   $X^{\lambda}$  is  an  SINR  graph  with  intensity  measure  
	$\lambda m:D \to [0,1]$ and   a   marked  probability  kernel  $Q$  from  $D$   to  $\R_+$  and  path  loss  function   $\ell(r)=r^{-\alpha}, $  for  $\alpha>0 .$ Assume $\lambda\big[\tau_{\lambda}(a)\gamma_{\lambda}(a)+\tau_{\lambda}(b)\gamma_{\lambda}(b)\big]\to \beta(a,b)\in(0,\infty)$,\, for  all  $a,b\in \R_+.$
	
	 Let  $Q$  be  the  exponential  distribution  with  parameter $c.$  Then, for  $\eps>0$   we  have  
	
	$$\lim_{\lambda\to\infty}\P\Big\{\sup_{(x,\sigma_x)\in\skrix}\Big|L_1^{\lambda}(x,\sigma_x)-m\otimes Q(x,\sigma_x) \Big|>\eps\Big\}=0$$ and  
		$$\lim_{\lambda\to\infty}\P\Big\{\sup_{([x,\sigma_x],[y,\sigma_y])\in\skrix\times\skrix}\Big|L_2^{\lambda}([x,\sigma_x],[y,\sigma_y])-e^{-R^{D}}m\otimes Q\times m\otimes Q.([x,\sigma_x],[y,\sigma_y]) \Big|>\eps\Big\}=0$$
\end{lemma}

\begin{proof}
	Let
	$$ F_1=\Big\{\omega:\sup_{(x,\sigma_x)\in\skrix}|\omega(x,\sigma_x)-m\otimes Q(x,\sigma_x)|>\eps\Big\},$$ $$F_2=\Big\{\varpi:\sup_{([x,\sigma_x],[y,\sigma_y])\in\skrix\times\skrix}|\varpi([x,\sigma_x],[y,\sigma_y])-e^{-R^{d,T}}m\otimes Q\times m\otimes Q([x,\sigma_x],[y,\sigma_y])|>\eps\Big \}$$
	
	and  $F_3=F_1\cup F_2.$     Now, observe  from  Theorem~\ref{main1a}  that
	
	$$\lim_{\lambda\to\infty}\frac{1}{\lambda}\log \P\Big\{(L_1^{\lambda},L_2^{\lambda})\in F_3^{c}\Big \}\le -\inf_{(\omega,\varpi) \in  F_3^{c}}I(\omega,\varpi).$$
	
	It  suffices  for  us  to  show  that   $I$  is strictly  positive.  Suppose  there  is  a  sequence  $(\omega_n,\varpi_n)\to(\omega,\varpi)$   such  that  $I(\omega_{\lambda},\varpi_{\lambda})\downarrow I(\omega,\varpi)=0.$  This  implies  $\omega=m\otimes Q$  and  $\varpi=e^{-R^{D}}m\otimes Q\times m\otimes Q$  which  contradicts  $(\omega,\varpi)\in F_3^{c}.$  This  ends  the  proof of  the  Lemma.
\end{proof} 

Now,  the  distribution  of  the  marked  PPP  $P(x)=\P\Big\{X^{\lambda}=x\Big\}$  is  given  by  $$P_{\lambda}(x)=\prod_{i=1}^{I}|\mu\otimes Q(x_i,\sigma_i)\prod_{(i,j)\in E}\frac{e^{-\lambda R_{\lambda}^{D}([x_i,\sigma_i],[y_j,\sigma_j])}}{1-e^{-\lambda R_{\lambda}^{D}([x_i,\sigma_i],[y_j,\sigma_j])}}\prod_{(i,j)\in \skrie} (1-e^{-\lambda R_{\lambda}^{D}([x_i,\sigma_i],[y_j,\sigma_j])})\prod_{i=1}^{I}(1-e^{-\lambda R_{\lambda}^{D}([x_i,\sigma_i],[y_i,\sigma_i])})$$

$$\begin{aligned}-\frac{1}{\lambda^2}\log P_{\lambda}(x)=\frac{1}{\lambda}\Big\langle -\log \mu\otimes Q\,,L_1^{\lambda}\Big\rangle +\Big \langle -\log \Big(\sfrac{e^{-\lambda R_{\lambda}^{D}}}{1-e^{-\lambda R_{\lambda}^{D}}}\Big ) \,,L_{2}^{\lambda}\Big \rangle & +\Big\langle-\log (1-e^{-\lambda R_{\lambda}^{D}})\,,L_1^{\lambda}\otimes L_{1}^{\lambda}\Big\rangle\\
& +\Big\langle-\log (1-e^{-\lambda R_{\lambda}^{D}})\,,L_{\Delta}^{\lambda}\Big\rangle
\end{aligned}  $$

Notice,  $\lim_{\lambda\to\infty}\lambda R_{\lambda}^{D}\to  R^{D},$       $\displaystyle \lim_{\lambda\to\infty}\frac{1}{\lambda}\Big\langle -\log \mu\otimes Q\,,L_1^{\lambda}\Big\rangle=\lim_{\lambda\to \infty}\frac{1}{\lambda}\Big\langle-\log (1-e^{-\lambda R_{\lambda}^{D}})\,,L_{\Delta}^{\lambda}\Big\rangle=0.$ 

 Using,  Lemma~\ref{WLLN}  we  have   $$\lim_{\lambda\to\infty}\Big \langle -\log \Big(\sfrac{e^{-\lambda R_{\lambda}^{D}}}{1-e^{-\lambda R_{\lambda}^{D}}}\Big ) \,,L_{2}^{\lambda}\Big \rangle=\Big \langle -\log \Big(\sfrac{e^{-R^{D}}}{1-e^{- R^{D}}}\Big ) \,, e^{-R^{D}}m\otimes Q\times m\otimes Q\Big \rangle$$

$$\lim_{\lambda\to\infty}\Big\langle-\log (1-e^{-\lambda R_{\lambda}^{D}})\,,L_1^{\lambda}\otimes L_{1}^{\lambda}\Big\rangle=\Big\langle-\log (1-e^{-R^{D}})\,,m\otimes Q\otimes \times m\otimes Q\Big\rangle,$$
which  concludes  the  proof  of  Theorem~\ref{main2}.

\section{Proof  of  Theorem~\ref{main2a} and  Corollary~\ref{main2b}}
For  $\omega\in\skrip(\skrix)$  we  define  the  spectral  potential  of  the  marked  SINR graph $(X^{\lambda})$  conditional  on  the  event  $\big\{L_1^{\lambda}=\omega \big\},$   $U_Q(g,\omega) $  as

\begin{equation}\label{LLDP.equ1}
U_{Q}(g,\omega )= \Big \langle g\,,\,e^{-R^{D}}\omega\otimes \omega \Big \rangle.
\end{equation}

The  following  remarkable  properties  holds  for  $U_{Q} $: 
\begin{itemize}
	\item  (i)  It  is  finite  on  $\skric(\omega ):=\Big\{g\in \skrix \to \R\,\Big |e^{U_Q(g,\omega )}< \infty \Big\}$
	\item (ii)  It  is  monotone. 
	\item (iii) it  is additively  homogeneous.  
	\item (iv) it  is  convex  in  $g.$
\end{itemize}

 For $\pi\in\skrim(\skrix\times\skrix)$, we  observe that  $I_{\omega}( \pi)$ is the  Kullback  action  of  the marked  SINR graph $X^{\lambda}$.
 
 \begin{lemma}\label {LLDP.equ2} The  following  hold  for the  Kullback  action or divergence function $I_{\omega}(\pi)$:
 	\begin{itemize}
 		\item $$I_{\omega}(\pi)=\sup_{g\in \skric}\big\{\langle g,\, \pi\rangle-\langle g,\, e^{-R^{D}}\omega\otimes \omega \rangle  \big\}$$ 
 		\item  The  function  $I_{\omega}(\pi) $  is  convex  and  lower semi-continuous   on  the  space  $\skrim(\skrix\times\skrix).$
 		\item For  any real  $\alpha$,  the  set  $\Big\{ \pi\in \skrim(\skrix\times\skrix): \, I_{\omega}(\pi)\le \alpha \Big\}$  is  weakly  compact.
 	\end{itemize}
 \end{lemma}
The  proof  of  Lemma~\ref{LLDP.equ2} is  omitted from  the  article. Interested  readers  may  refer  to  \cite{DA2017a}  for  similar proof for empirical  measures of ` the  Typed  Random  Graph  Processes,  and/or the  references therein  for proof of  the  lemma for  empirical measures  on  measurable  spaces.\\

Now   note  from  Lemma~\ref{LLDP.equ2},  for  any  $\eps>0$, there  exists a  function  $g\in\skrim(\skrix\times\skrix)$   such  that  
$$I_{\omega}(\pi)-\sfrac{\eps}{2}< \langle g\,,\, \pi \rangle -U_{Q}(g,\omega).$$

Define  the  probability  distribution   $P_{\omega} $  by  

$$P_{\omega}(x)=\prod_{(i,j)\in E}e^{g(x_i,x_j)}\prod_{(i,j)\in \skrie}e^{h_\lambda(x_i,x_j)}	,$$

where  $$h_{\lambda}(x,y)=\lambda \log\Big[(1-e^{-\lambda R_{\lambda}^{D}(x,y)}+e^{-\lambda R_{\lambda}^{D}(x,y)+g(x,y)/\lambda})\Big]$$
Then,  observe  that  
                    
  $$\begin{aligned}
\frac{dP_{\omega}}{d\tilde{P}_{\omega}}(x)&=\prod_{(i,j)\in E}e^{-g(x_i,x_j)/\lambda}\prod_{(i,j)\in \skrie} e^{-h_{\lambda}(x_i,x_j)/\lambda}\\
&= e^{-\lambda (\langle\sfrac{1}{2} g,L_2^{\lambda}\rangle-\lambda \langle \sfrac{1}{2}h_{\lambda},L_1^{\lambda}\otimes L_1^{\lambda}\rangle )+\langle \sfrac{1}{2} h_{\lambda}, L_{\Delta}^{\lambda} \rangle }
 \end{aligned}$$

 Now,  we  define  the  neighbourhood  of  $\nu,$    $B_{\nu}$ by
 $$B_{\nu}:=\Big\{\pi\in \skrim(\skrix\times\skrix): \, \langle g, \pi \rangle >\langle g, \nu\rangle -\eps/2 \Big \}$$

 Observe, under the  condition  $L_{2}^{\lambda}\in B_{\nu}$   we  have

   $$\begin{aligned}
 \frac{dP_{\omega}}{d\tilde{P}_{\omega}}< e^{-\langle\sfrac{1}{2} g,\nu\rangle+U_{Q}(g,\,\omega)+\lambda \sfrac{\eps}{2}}<e^{-\lambda  I_{\omega}(\nu)+\lambda \eps} \end{aligned}$$
 
 Hence,  we  have  
 
$$P_{\omega }\Big\{ x^{\lambda }\in\skrig_P\Big | L_2^{\lambda}\in B_{\nu}\Big\}\le \int \1_{\{L_2^{\lambda}\in B_{\nu}\}} d\tilde{P}_{\omega}(x^{\lambda})\le \int e^{-\lambda I_{\omega(\nu)}-\lambda\eps}d\tilde{P}_{\omega}(x^{\lambda}) \le e^{-\lambda I_{\omega}(\nu)-\lambda\eps}.$$.

Observe  that   $I_{\omega} (\nu)=\infty$ implies  Theorem ~\ref{main2a} (ii),  hence  it  sufficient for  us  to  establish  it  for  a probability  measure  of  the  form $\nu=ge^{-R^{D}}\omega\otimes \omega,$ where  $g=1$  and  for  $I_{\omega}(\nu)=0.$  Fix  any  number  $\eps>0$  and  any  neigbourhood $B_{\nu}\subset \skrim(\skrix\times\skrix)$.  Now  define  the  sequence  of  sets  $$\skrig_{P}^{\lambda}=\Big\{ y\in \skrig_{P}: L_{2}^{\lambda}(y)\in B_{\nu}, \Big |\langle g, L_{2}^{\lambda}\rangle-\langle g, \nu\rangle\Big |\le \sfrac{\eps}{2}\Big\} .$$     

 Note that  for  all  $y\in \skrig_P^{\lambda}$ we  have    
 
  $$\begin{aligned}
\frac{dP_{\omega}}{d\tilde{P}_{\omega}}> e^{-\langle\sfrac{1}{2} g,\nu\rangle+U_{Q}(g,\,\omega)+\lambda \sfrac{\eps}{2}}>e^{\lambda \eps} \end{aligned}.$$

This  yields  
$$P_{\omega}(\skrig_P^{\lambda}) =\int_{\skrig_P^{\lambda}}dP_{\omega}(y)\ge\int e^{-\langle\sfrac{1}{2} g,\nu\rangle+U_{Q}(g,\,\omega)+\lambda \sfrac{\eps}{2}}d\tilde{P}_{\omega}(y)\ge e^{\lambda \eps}\tilde{P}_{\omega}(\skrig_P^{\lambda}).$$
Using   the law  of  large  numbers, we  have  that  $\lim_{\lambda\to\infty}\tilde{P}_{\omega}(\skrig_P^{\lambda})=1.$  This  completes  of  the  Theorem.\\

{\bf  Proof of  Corollary~\ref{main2b}}

We observe that, by  Lemma~\ref{Tigtness} the  law  of   empirical connectivity measure is exponentially  tight. Henceforth, without loss of generality we can assume that the
set $F$  in Theorem 2.2(ii) above is relatively compact. If we choose any $\eps> 0$; then for each functional  $\nu\in F$ we can find a weak neigbourhood such that the estimate of Theorem 2.1(i) above holds. From
all these neigbhoourhood, we choose a finite cover of $\skrig_P$ and sum up over the estimate in
Theorem 2.1(i) above to obtain   	$$\limsup_{\lambda\to\infty}\frac{1}{\lambda}\log \P_{\omega}\Big\{X^{\lambda}\in \skrig_P\,\Big|\, L_2 ^{\lambda}\in F\Big\}\le -\inf_{\pi\in F}I_{\omega}(\pi)+\eps.$$

Since $\eps$  was arbitrarily chosen and the lower bound in Theorem 2.1(ii) is implies the lower bound in
Theorem 2.2(i) we have the required results which completes the proof.


\end{document}